\documentclass[conference]{IEEEtran}

\ifCLASSINFOpdf
   \usepackage[pdftex]{graphicx}
   \graphicspath{{./}}
   \DeclareGraphicsExtensions{.pdf,.jpeg,.png,.eps}
\else
\fi

\usepackage[cmex10]{amsmath}
\usepackage{booktabs} 
\usepackage{multirow} 
\usepackage{multicol}
\usepackage{balance}
\usepackage{cite}
\usepackage[switch]{lineno} 
\usepackage{enumitem}
\usepackage{dsfont}
\usepackage{amssymb}
\usepackage{amsthm}
\usepackage{tablefootnote}
\usepackage{colortbl}

\newtheorem{prop}{Proposition}

\usepackage[bookmarks=false]{hyperref}
\hypersetup{
  unicode=false,          
  pdftoolbar=true,        
  pdfmenubar=true,        
  pdffitwindow=false,     
  pdfstartview={FitH},    
  pdftitle={My title},    
  pdfauthor={Jun Xing Chin},     
  pdfsubject={Subject},   
  pdfcreator={Creator},   
  pdfproducer={Producer}, 
  pdfkeywords={keyword1} {key2} {key3}, 
  pdfnewwindow=true,      
  colorlinks=true,       
  linkcolor=black,          
  citecolor=black,        
  filecolor=black,      
  urlcolor=black           
}

\usepackage{fancyhdr}
\makeatletter

\usepackage{algorithm}
 \usepackage{algpseudocode}
 
\usepackage[caption=false,font=footnotesize]{subfig}

\newcommand{\tilI}{\tilde{I}}

\usepackage[]{changes}

\setdeletedmarkup{\color{red}{ \sout{#1}}}
\setaddedmarkup{\textit{#1}}

\begin{document}

\title{Consumer Privacy Protection using\\ Flexible Thermal Loads:\\ Theoretical Limits and Practical Considerations}

\author{
\IEEEauthorblockN{Jun-Xing Chin\IEEEauthorrefmark{1}, Kyri Baker\IEEEauthorrefmark{2}, and Gabriela Hug\IEEEauthorrefmark{3}}\\%
\IEEEauthorblockA{\IEEEauthorrefmark{1} Singapore-ETH Centre, ETH Zurich, Singapore}%
\IEEEauthorblockA{\IEEEauthorrefmark{2} Architectural Engineering, University of Colorado Boulder, Colorado, USA}%
\IEEEauthorblockA{\IEEEauthorrefmark{3} Power Systems Laboratory, ETH Zurich, Zurich, Switzerland}%

Emails: \{chin, hug\}@eeh.ee.ethz.ch, Kyri.Baker@colorado.edu   
}

\fancypagestyle{firstpage}{
	\fancyhf{}
	\renewcommand\headrule{}
	\setlength{\voffset}{-0.4cm}
	\fancyhead[C]{\scriptsize This is the author's version of an article that has been published in this journal. Changes were made to this version by the publisher prior to publication.\\
		The final version of record is available at~ \href{https://doi.org/10.1016/j.apenergy.2020.116075}{\color{blue}{https://doi.org/10.1016/j.apenergy.2020.116075}}\\~\\}
	\fancyhead[L]{\footnotesize Applied Energy}
	\fancyfoot[CO]{\scriptsize~\\~\\~\\Copyright (c) 2020. This manuscript version is made available under the CC-BY-NC-ND 4.0 license \href{http://creativecommons.org/licenses/by-nc-nd/4.0/}{\color{blue}{http://creativecommons.org/licenses/by-nc-nd/4.0/}}}
}

\IEEEoverridecommandlockouts

\pagestyle{fancy}
\renewcommand\headrule{}
\setlength{\voffset}{-0.4cm}
\fancyhead[C]{\scriptsize This is the author's version of an article that has been published in this journal. Changes were made to this version by the publisher prior to publication.\\
	The final version of record is available at~ \href{https://doi.org/10.1016/j.apenergy.2020.116075}{\color{blue}{https://doi.org/10.1016/j.apenergy.2020.116075}}\\~\\}
\fancyhead[L]{\footnotesize Applied Energy}
\fancyfoot[CO]{~\\\scriptsize Copyright (c) 2020. This manuscript version is made available under the CC-BY-NC-ND 4.0 license \href{http://creativecommons.org/licenses/by-nc-nd/4.0/}{\color{blue}{http://creativecommons.org/licenses/by-nc-nd/4.0/}}}

\maketitle
\thispagestyle{firstpage}

\begin{abstract}
The increasing adoption of smart meters introduces growing concerns about consumer privacy risks stemming from high resolution metering data. To counter these risks, there have been various works in actively shaping the grid-visible energy consumption profile using controllable loads such as energy storage systems (ESSs) and flexible consumer loads. In this paper, we compare the use of flexible thermal-based consumer loads (FTLs) against ESSs for consumer privacy protection. By first assuming ideal conditions, and subsequently bringing them closer to reality, the limitations of using FTLs for privacy protection are identified. Through theoretical analyses and realistic simulations, it is shown that, due to the limitations in the operation of FTLs, without significant over-sizing of systems and sacrifices in consumer comfort, FTLs of much higher equivalent energy storage capacity are required to afford the same level of protection as ESSs. Nonetheless, given their increasing ubiquity, controllable FTLs should be considered for use in consumer privacy protection.
\end{abstract}


\section{Introduction}\label{Sec:intro}
Spurred by grid modernisation efforts, the adoption rate of advanced metering infrastructure (AMI) using smart meters (SMs) has risen steadily across the globe in recent years. On one hand, this enables the development of efficient data-driven grid operation and management methods \cite{McKenna2012}. On the other hand, the high-frequency measurement data provided by the AMI can be used to derive private information of consumers, such as their lifestyle habits, occupation, and religious inclinations \cite{McDaniel2009, Molina-Markham2010, McKenna2012}. 
The authors in \cite{McKenna2012} provide a comprehensive overview of applications (and information) that can be derived from SM data, while in \cite{Molina-Markham2010}, the authors explore the granularity of SM measurements required to infer specific household activities, and show that some private information can still be inferred at an hourly resolution. More importantly, the authors of \cite{McDaniel2009} find that the existing laws in the US are unclear regarding customer energy data usage, which potentially paves the way for its exploitation. Moreover, a 2017 survey in the US has shown that utilities pose high privacy risks, and are not highly trusted by consumers \cite{pwc2017}. Even in the presence of clear laws that prevent the exploitation of SM data by utility companies, such as the European Union's General Data Protection Regulation \cite{EUGDPR16}, the underlying metering infrastructure is still vulnerable to cyber-attacks, which may lead to SM data disclosure to malicious adversaries \cite{McLaughlin2010}. 

This has led to concerns regarding privacy risks \cite{Winter2019}, and push-backs against the use of SMs, delaying and potentially altering the scope of their deployment, e.g., in the Netherlands \cite{VanAubel2019}. These concerns have motivated works in quantifying and mitigating these risks, such as \cite{Sankar2013,McLaughlin2011,Yang2012,Tan2017,Chin2017,Zhang2017,BakerNILM,Mashima2018,Arzamasov2020}. Nonetheless, quantifying privacy in a meaningful manner remains an ongoing research challenge \cite{Arzamasov2020}. In \cite{BakerNILM} and \cite{Mashima2018}, the authors use the performance of specific data analytics applications, \emph{i.e.}, non-intrusive load monitoring (NILM) and socio-demographic classifiers, respectively, as measures of consumer privacy. While they allow simple and meaningful interpretation of a consumer's level of privacy, NILM techniques, which are reliant on appliance load signatures, are very sensitive to small perturbations in the data \cite{BakerNILM}; whereas machine learning-based classifiers are sensitive to rudimentary privacy protection measures \cite{Mashima2018}. This makes them less robust as a measure of privacy against more sophisticated adversaries. On the other hand, information theoretic privacy measures, such as mutual information (MI)\cite{Sankar2013} and differential privacy\cite{Zhang2017} are attack-agnostic and offer more robust privacy guarantees, but are less readily interpretable in a meaningful manner \cite{Arzamasov2020}. 

A recent article by Giaconi et al.\cite{Giaconi2018a} provides a high-level overview of privacy protection methods for consumers with SMs. In particular, privacy protection schemes can be categorised into two main families, namely \textit{smart meter data manipulation} (SMDM) schemes, and \textit{user demand shaping} (UDS) schemes \cite{Giaconi2018a}. SMDM schemes modify the SM data before it is transmitted, and include aggregating SM measurements before transmission to the data collector \cite{Danezis2013, Koo2017, Mustafa2019}; anonymising the SM measurements to decouple SM data from individual households \cite{Efthymiou2010, Rottondi2012}; and the differential privacy-based addition of noise \cite{Ni2017,Acs2011,Hassan2019}. However, these methods require trusted third parties, either in the processing of the data, or in the supply and installation of SMs with privacy-preserving firmware. UDS methods, on the other hand, physically alter the physical energy consumption profiles of consumers recorded by the SMs (\textit{grid load}), such that they no longer reveal the private information contained in the underlying privacy-sensitive consumer load profiles (\textit{sensitive load}). This is achieved by actively controlling loads to shape the grid load profile, ideally decoupling it from the sensitive load profile. UDS methods can typically be implemented behind-the-meter, which avoids the need for a trusted third-party. 

UDS methods can be further classified into those using energy storage systems (ESSs), those controlling flexible consumer loads, and those using a combination of the two. Fig. \ref{fig:SysMod} illustrates a possible system setup for UDS methods, which is governed by the equation:

\noindent\small
\begin{equation} \label{eq:powerbalance}
    Y = X + S.
\end{equation}\normalsize 
Hence, the flexibility of controllable loads $S$, such as ESSs and flexible consumer loads, is used to influence what can be derived from the grid load $Y$ about the sensitive load $X$.

There are numerous recent UDS schemes that only use ESSs (also known as battery load hiding), e.g., load levelling \cite{McLaughlin2011}, limiting the load profile to distinct steps \cite{Yang2012}, and directly minimising an approximate of MI \cite{Chin2017}. In \cite{Giaconi2018b}, the authors derive theoretical privacy guarantees for consumers with ESSs and renewable energy sources based on ideal assumptions, and show that, while it is possible to numerically evaluate the privacy bounds for realistic batteries using the Blahut-Arimoto algorithm, it is computationally intractable in practice. Arzamasov et al. provide a more recent overview of SM related privacy measures for ESS-based UDS methods in their recent work \cite{Arzamasov2020}, where they also found that the choice of privacy metrics and the characteristics of a consumer's load profile greatly affect the relative performance of ESS-based UDS schemes. They argue that an ideal privacy measure would be the reconstructability of the original unprotected consumer load profile. However, assessing the reconstructability of the consumer load profile given a specific privacy protection scheme is a non-trivial problem that remains to be solved.

On the other hand, UDS methods utilising flexible consumer loads are scarce in the literature. One such UDS scheme, proposed in \cite{Chen2015}, utilises the flexible consumer loads to hide occupancy by using artificial signature injection and partial load flattening. The authors then verify their scheme by testing the resultant load profiles using a few occupancy detection algorithms. Another flexible consumer load-based UDS scheme is given in \cite{Sun2018}, where the authors use flexible consumer loads aided with batteries for privacy protection. However, no in-depth assessment or discussion on the performance of the proposed scheme is included. In \cite{BakerNILM}, optimised electric vehicle charging and an electric furnace are used to obscure recoverable information from NILM techniques. Notwithstanding, the use of flexible consumer loads for general privacy protection irrespective of the adversarial model, and their performance against schemes based on ESSs, are not well studied. 
\begin{figure}
    \centering
    \includegraphics[trim=3.5cm 8cm 11.5cm 3cm, clip=true, width=0.95\columnwidth]{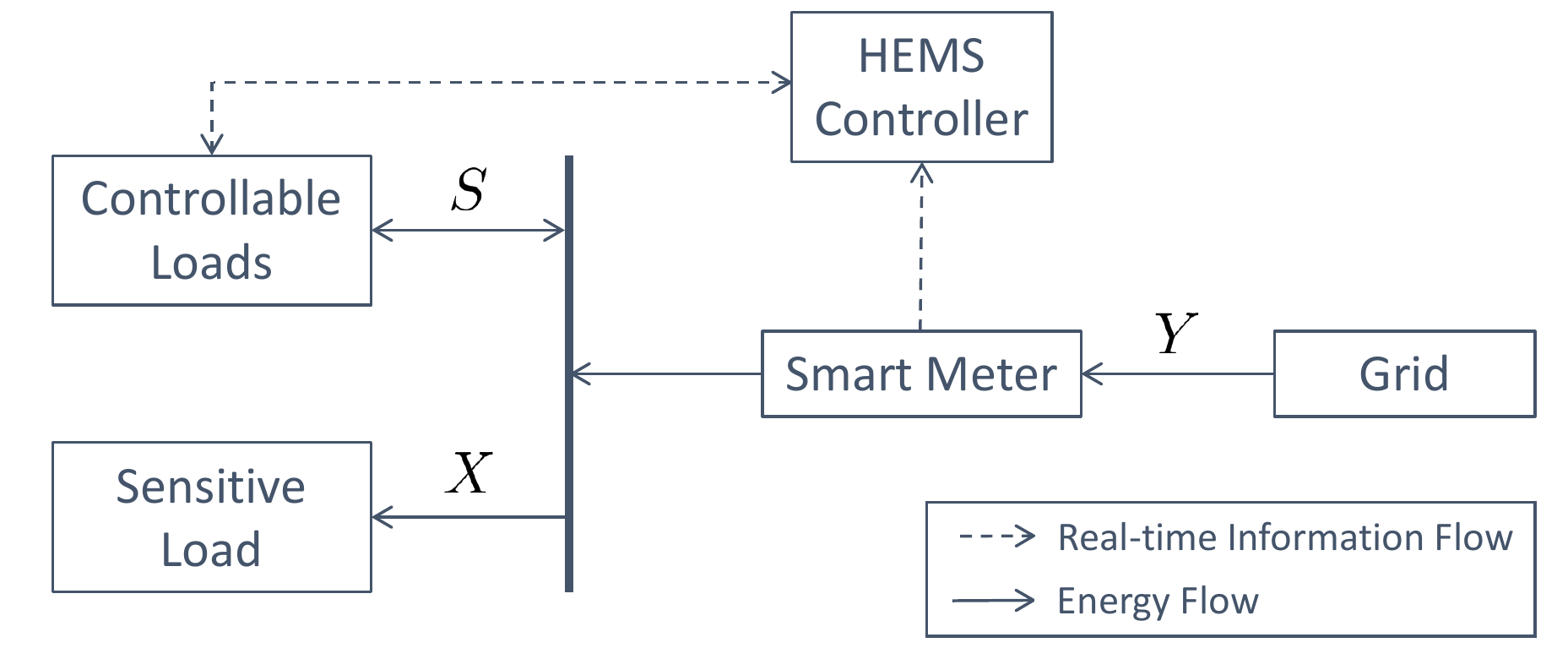}
    \vspace{-0.1cm}
    \caption{Possible UDS system setup, where the controllable load is controlled by the HEMS controller to actively shape the grid visible load to mask the private information in the privacy-sensitive load.}
    \label{fig:SysMod}
    \vspace{-0.5cm}
\end{figure}

With the development of grid communications infrastructure and the proliferation of smart appliances, there are also considerable advances in home energy management systems (HEMSs) that enable the coordination and scheduling of home appliances. HEMSs allow for the optimisation of residential electricity consumption patterns in order to improve efficiency, economics, and the reliability of residential buildings with regards to their role in the grid and occupant comfort\cite{Zhou2016}. Given increasing interest in HEMSs and the ubiquity of flexible consumer loads, this paper explores the use of HEMS-controlled flexible consumer loads in order to mask the private information contained in the grid load about the sensitive load. Specifically, the contributions of this paper are three-fold:
\begin{itemize}
    \item The concept of privacy and flexible consumer loads is formalised for households with smart meters.
    \item The theoretical limits of privacy protection using ESSs and flexible thermal-based consumer loads are analysed, with the findings validated for real-life applications using realistic numerical simulations.
    \item To the best of our knowledge, this paper is the first work to directly compare consumer privacy protection of systems using ESSs against those using flexible thermal-based consumer loads.
\end{itemize}

The rest of this paper is structured as follows: Section \ref{Sec:PrivacyMeasure} provides a brief overview of quantifying privacy loss for consumers with smart meters; Section \ref{Sec:FCLandPrivacy} briefly discusses the use of flexible consumer loads for privacy protection; Section \ref{Sec:FTL-ESS_Theory} provides an analytical comparison between consumer privacy protection using ESSs and flexible thermal-based consumer loads; Section \ref{Sec:ControllerSetup} details the controller design of a HEMS for comparison of realistic systems; Section \ref{Sec:NumEx} presents numerical results; and Section \ref{Sec:Conc} concludes the paper.  

\section{Quantifying Consumer Privacy Loss}\label{Sec:PrivacyMeasure}
As previously mentioned, one measure of consumer privacy loss is the mutual information between the sensitive load $X$ and the grid load $Y$ \cite{Sankar2013, Tan2017}, which measures the amount of information $Y$ reveals about $X$ and vice versa. Mutual information is capable of modelling nonlinear relationships between variables, unlike using correlation coefficients, for example. The MI between $X$ and $Y$, which are random processes, can be given as the average MI between the random variables $X_{\tau}$ and $Y_{\tau}$ that make up the processes \cite{Tan2017,Chin2018}, \emph{i.e.}, 

\noindent \small
\begin{equation}
    I(X;Y) = \frac{1}{k} \sum_{\tau = 1}^{k} I(X_{\tau};Y_{\tau}) ~, 
\end{equation}
\normalsize
where $I(X_{\tau};Y_{\tau})$ is the MI between the random variables $X_{\tau}$ and $Y_{\tau}$, and $k$ is the number of random variable pairs. This concept of average MI will be used in Section \ref{Sec:FTL-ESS_Theory} for the analysis of consumer privacy protection.

Given two random variables $X_{\tau}$ and $Y_{\tau}$, the MI between them is given by a function of their joint probability distribution function (PDF) $p_{X_{\tau},Y_{\tau}}$, and marginal distributions, $p_{X_{\tau}}$, and $p_{Y_{\tau}}$. These PDFs are typically unknown, and must be estimated. Assuming that multiple samples of $X_{\tau}$ and $Y_{\tau}$ are available, the PDFs can be estimated using the histogram method. Hence, only for the purpose of estimating these PDFs, assume that the protected and grid loads have finite support, \emph{i.e.}, $X_{\tau} \in \mathcal{X}_{\tau} := \{\bar{x}^1,\bar{x}^2,\cdots,\bar{x}^m\}$, and $Y_{\tau} \in \mathcal{Y}_{\tau} := \{\bar{y}^1,\bar{y}^2,\cdots,\bar{y}^n\}$. Then, the MI between $X_{\tau}$ and $Y_{\tau}$ can be given as 

\noindent \small 
\begin{equation}\label{eq:discMI}
    I(X_{\tau};Y_{\tau}) := \sum_{i=1}^{m} \sum_{j=1}^{n} p_{X_{\tau},Y_{\tau}}(\bar{x}_{}^{i},\bar{y}_{}^{j}) \log \frac{p_{X_{\tau},Y_{\tau}}(\bar{x}_{}^{i},\bar{y}_{}^{j})}{p_{X_{\tau}}(\bar{x}_{}^{i})p_{Y_{\tau}}(\bar{y}_{}^{j})},
\end{equation} 
\normalsize
where $p_A(a)$ denotes the probability of $A = a$, and $\log$ is the base-2 logarithm. For the rest of the paper, we further denote the realisations of the random variables with lowercase letters, $A^{\text{mean}}$ as the average value of $A$, $A^{\text{min}}$ as the minimum value that $A$ can take, and $A^{\text{max}}$ as the maximum value of $A$. 

As $X$ and $Y$ are continuous in reality, the PDF estimates become more accurate with an increase in $m$ and $n$; but this also requires more samples to prevent over-fitting. It follows that in order to minimise leakage of privacy-sensitive information, one needs to minimise the MI between the sensitive and grid loads. This can be done either through UDS or SMDM methods as described in Section \ref{Sec:intro}; and for UDS methods, using either ESSs, flexible consumer loads, or a combination thereof.

\section{Flexible Consumer Loads \\ and Consumer Privacy}\label{Sec:FCLandPrivacy}
The term ``flexible consumer loads" include thermal loads such as hot water heaters and space conditioning, schedulable loads such as clothes and dishwashers, and interruptible loads such as the charging of electric vehicles. From a privacy perspective, flexible consumer loads can broadly be classified into the following categories:
\begin{enumerate}[label=\alph*)]
    \item  Flexible consumer loads that are not privacy-sensitive, \emph{i.e.}, their usage does not reveal privacy-sensitive information about the consumer, nor are their presence in a household considered sensitive private information; e.g., electric space heaters within a house with high thermal inertia, in a community where their presence is the norm.
	\item  Flexible consumer loads that are privacy-sensitive with regards to their time-of-use, but not their presence in the household; e.g., electric stoves in a community where their presence is the norm. 
	\item  Flexible consumer loads that are privacy-sensitive, \emph{i.e.}, both their time-of-use and presence in a household reveal sensitive private information; e.g., electric stoves in a community where households typically cook with gas stoves.
\end{enumerate}
When controlling flexible consumer loads to shape a user's demand and reduce their information leakage, the privacy sensitivity of the loads themselves need to be considered. There are no privacy issues arising from their usage if the flexible consumer loads are of the first category. For loads of the second category, using them to mask the sensitive load inherently also masks the private information they reveal: their time-of-use is shifted and thus, the private information revealed by their original time-of-use is masked. However, if the flexible consumer loads are of the third category, then the privacy-protection problem also needs to consider whether the resulting grid load is able to mask the electrical signature of the flexible consumer loads, \emph{i.e.}, whether the sensitive load is able to sufficiently distort the signatures of the flexible consumer loads as well.

To simplify the analysis, we consider the use of flexible consumer loads within the first two categories in UDS privacy-protection schemes in this paper. Moreover, we limit our analysis to flexible thermal-based consumer loads (flexible thermal loads or FTLs) due to their ability to `store' thermal energy, and are more likely to be interruptible compared to other types of flexible consumer loads, such as washing machines that have minimum cycle times. Inductive FTLs, such as heat pumps, have complex on/off cycles and electrical signatures, making the analysis of their effectiveness in privacy protection complicated. Hence, in order to draw meaningful conclusions, we will focus on resistance-based FTLs, such as electric-resistance water heaters, and electric-resistance space heaters. In the next section, we will compare the theoretical performance of ESS-based UDS schemes against those using resistance-based FTLs.

\section{Comparing Privacy Protection using Energy Storage Systems and Flexible Thermal Loads}\label{Sec:FTL-ESS_Theory}
Setting aside the distinctive constraints of both ESSs and FTLs, the privacy protection afforded by them for UDS differs in one key aspect: ESSs are able to both charge and discharge, \emph{i.e.}, increase or decrease grid load; while traditional residential FTLs are only able to `charge', \emph{i.e.}, they can draw power from the grid, but typically cannot provide power back to the grid. 

Let $H(\cdot) := - \sum p(\cdot) \log p(\cdot)$ be the Shannon entropy function, with $p(\cdot)$ being the probability of the variable and $H(\cdot)$ being minimal when the outcome is certain, and maximal when the underlying distribution is uniform. Additionally, assume that the following is true: 
\begin{enumerate}[label=(\alph*)]
    \item No energy wastage is permitted.
    \item The power ratings of the ESS and FTL are sufficiently large to compensate for the difference between the maximum and minimum consumer load, \emph{i.e.}, $P_{\text{ess}}^{\text{max}} , P_{\text{th}}^{\text{max}} \geq X^{\text{max}}-X^{\text{min}}$. 
    \item The controller has perfect knowledge of the efficiency curves of the ESS and FTL.
    \item The controller has perfect knowledge of the consumer load $X$ and its average $X^{\text{mean}}$.\label{as:unreal1}
    \item The ESS has infinite energy storage capacity.\label{as:SysCap1}
    \item Either the FTL has infinite thermal storage capacity, or it holds, for the average electrical equivalent of the consumer thermal demand $D_{\text{th}}^{\text{mean}}$, that $D_{\text{th}}^{\text{mean}} \geq P_{\text{th}}^{\text{max}}$.\label{as:SysCap2}
    \item The FTL demand is continuous, \emph{i.e.}, it is not a step-load.\label{as:unreal2}
    \item Both ESS and FTL have an initial state-of-charge of $0.5$.
\end{enumerate} 
Using MI as the measure of privacy, the differences in achievable privacy protection by both technologies are discussed in the remainder of this section. 

\subsection{The Loads are Independent and Identically Distributed}
Let the random variable pair $(X,Y)$, and its marginals $X$ and $Y$ be independent and identically distributed (i.i.d.). Then, by definition, MI, $I_\text{iid}(X;Y)$ can also be written as a function of their Shannon entropies,

\noindent \small
\begin{alignat}{2}
I_\text{iid}(X;Y) 
&= &&\, H(X) + H(Y) - H(X,Y) \label{FTLeq:MIEntropy} \\
&= &&\, H(Y) - H(Y|X) ~. \label{FTLeq:MIEntropy2}
\end{alignat}
\normalsize
From \eqref{FTLeq:MIEntropy}, it is trivial to see that $I_\text{iid}(X;Y)$ can be minimised by either maximising $H(X,Y)$, assuming $H(X,Y)$ increases at the same or higher rate than $H(Y)$; or by minimising $H(Y)$, assuming $H(Y)$ decreases at the same or higher rate than $H(X,Y)$. Note that as $H(X)$ is fixed, $H(Y)$ and $H(X,Y)$ are affected similarly (either both increase or both decrease) by a given control action. Moreover, we have the following propositions:

\begin{prop} \label{Prop:MinMI}
	$I_\text{iid}(X;Y)$ is minimal when $H(Y)$ is minimal, \emph{i.e.}, when $|\mathcal{Y}' := \{y \in \mathcal{Y} ~|~ p_Y(y) >0 \} |$ is minimal.
\end{prop}
\begin{proof} \label{Proof:MinMI}
	The distribution of the sensitive load $p_X(x)$ is uncontrollable, non-uniform, and the number of outcomes with non-zero probability is non-singular, \emph{i.e.}, $|\mathcal{X}' := \{x \in \mathcal{X} ~|~ p_X(x) >0 \} | > 1$. Hence, $H(X)$ is greater than zero. Since $p_{X,Y}(x,y)$ is non-uniform, as $p_X(x)$ is non-uniform, $H(X,Y)$ is limited by the given $p_X(x)$. Therefore, it follows that $I_\text{iid}(X;Y)$ is minimal when $H(Y)$ is minimal, \emph{i.e.}, when $|\mathcal{Y}'|$ is minimal, where $\mathcal{Y}' := \{y \in \mathcal{Y} ~|~ p_Y(y) >0 \}$, instead of when $H(X,Y)$ is maximal. 
\end{proof}

\begin{prop} \label{Prop:MinMI2}
	$I_\text{iid}(X;Y)$ is minimal when $H(Y|X)$ is maximal.
\end{prop}
\begin{proof} \label{Proof:MinMI2}
	Given a fixed and non-uniform sensitive load distribution $p_X(x)$, $H(Y|X)$ is maximal when $p_{Y|X}(y|x)$ are uniform distributions for each value of $x$. Since $p_{Y}(y) = \sum_\mathcal{X} p_{X,Y}(x,y) = \sum_\mathcal{X} p_{Y|X}(y|x)p_X(x)$, it is also a uniform distribution when $p_{Y|X}(y|x)$ are uniform distributions for each value of $x$. Therefore, $H(Y) = H(Y|X)$ when $H(Y|X)$ is maximal, and $I_\text{iid}(X;Y) = 0$, which is its minimal value. 
\end{proof}

As FTLs cannot `discharge' (reduce the grid load), and therefore, cannot achieve a uniform distribution for $p_{Y|X}(y|x)$, the following analysis is based on Proposition \ref{Prop:MinMI}. It is trivial to see that perfect privacy, $I_\text{iid}(X;Y) = 0$ can be achieved by maintaining a constant grid load, $y^*$, where $p_Y(y^*) = 1~, \text{and } p_Y(y) = 0 ~\forall y \neq y^*$. Let the grid load achieved using the ESS be denoted by $Y_{\text{ess}}$ and that of FTL by $Y_{\text{th}}$, then there exists $y^*_{\text{ess}}$ and $y^*_{\text{th}}$ such that $I_\text{iid}(X;Y^*_{\text{ess}}) = I_\text{iid}(X;Y^{*}_{\text{th}}) = 0$. While $y_{\text{ess}}^*$ can be any arbitrary value $Y^{\text{min}} \leq y_{\text{ess}}^* \leq Y^{\text{max}}$, there is less flexibility for $y_{\text{th}}^*$, with $y_{\text{th}}^* \geq X^{\text{max}}$. Nonetheless, the theoretical maximum privacy can be achieved by both technologies given ideal assumptions.

In reality, storage capacity is finite, and for most consumers, it would be unreasonable to assume that the system is undersized, \emph{i.e}, $D_{\text{th}}^{\text{mean}} \geq P_{\text{th}}^{\text{max}}$, where $D_{\text{th}}^{\text{mean}}$ is the electrical equivalent of the average power consumption required in order to maintain consumer comfort. Therefore, assumptions \ref{as:SysCap1} and \ref{as:SysCap2} are made more realistic such that the storage capacity is finite, but sufficiently large to average out consumer load (or thermal demand) over a finite period of time. Additionally, average thermal demand is now assumed to be large, but less than the FTL power rating and that $D_{\text{th}}^{\text{mean}} + X^{\text{mean}} < X^{\text{max}}$. For ESSs, the controller would now need to select a constant grid load such that $y_{\text{ess}}^* = X^{\text{mean}} + l_{\text{ess}}$, where $l_{\text{ess}}$ is the round trip loss of the ESS. This allows a constant $y^*_{\text{ess}}$ that does not empty or fully charge the ESS. As it would be possible to sustain $y^*_{\text{ess}}$ indefinitely, $I_\text{iid}(X;Y_{\text{ess}}) = I_\text{iid}(X;Y^*_{\text{ess}}) = 0$. For FTLs, it follows that $y_{\text{th}}^* = D_{\text{th}}^{\text{mean}} + X^{\text{mean}}$, and that $y_{\text{th}}^* < X^{\text{max}}$ in most realistic cases. Assume that $I_\text{iid}(X;Y_{\text{th}})$ is still minimised by actuating $y_{\text{th}}^*$ whenever possible. In this case, we now also have $y = x \neq y_{\text{th}}^* ~, \forall x > y_{\text{th}}^*$. Let $k$ be the total number of samples, and $g(k)$ be the number of instances where $x > y_{\text{th}}^*$, then 

\vspace{-0.3cm}\small
\begin{subequations}
\begin{alignat}{2}
    I_\text{iid}(X;Y_{\text{th}}) &= \,&&\frac{k-g(k)}{k} I(X;Y^*_{\text{th}}) + \frac{g(k)}{k} H(X)\\
                &= \,&& \frac{g(k)}{k} H(X) ~, \label{eq:FTLPrivacy}
\end{alignat}
\end{subequations}
\normalsize
where \eqref{eq:FTLPrivacy} follows from the fact that $I_\text{iid}(X;Y^*_{\text{th}}) = 0$ and $y = x ~, \forall x > y_{\text{th}}^*$. Thus, $I_\text{iid}(X;Y_{\text{ess}}) < I_\text{iid}(X;Y_{\text{th}})$ as $H(X) > 0$, \emph{i.e.}, privacy loss using FTLs for UDS schemes is, under the given assumptions on equivalent storage size, greater than those using ESSs given these assumptions. 

In theory, MI can also be minimised by maximising $H(Y|X)$ if $|\mathcal{Y}' := \{y \in \mathcal{Y} ~|~ p_Y(y) >0 \} |$ is too large, \emph{i.e}, when there are too many different values of $x > y_\text{th}^*$. However, given that FTLs cannot `discharge', $p_{Y|X}(y|x)$ cannot be uniform distributions. Therefore, $H(Y)-H(Y|X) > 0$, as $H(Y) \neq H(Y|X)$ in this case, again, resulting in $I_\text{iid}(X;Y^*_\text{th}) > I_\text{iid}(X;Y^*_\text{ess})$.

\subsection{The Loads are First-Order Markov Processes}
The random variables $(X,Y)$, $X$, and $Y$ are not i.i.d. in reality, and could be better modelled using first-order Markov processes \cite{McLoughlin2010}, of which the MI, $I_\text{m}(X;Y)$ \cite{Tan2017} is given by

\noindent \small
\begin{alignat}{2}
I_\text{m}(X;Y) 
&= ~&& \frac{1}{k} \left[ \sum_{\tau=2}^{k} I(X_{\tau},X_{\tau-1};Y_{\tau},Y_{\tau-1}) - \right. \notag\\ 
&  ~&& \quad\quad\quad\quad\quad \left. \sum_{\tau=3}^{k} I(X_{\tau-1},Y_{\tau-1})\right]~. \label{eq:MarkovMI}
\end{alignat} \normalsize
Expressing \eqref{eq:MarkovMI} in terms of entropy,

\vspace{-0.3cm}\small \begin{alignat*}{2}
I_\text{m}(X;Y) 
&= ~&& \frac{1}{k}  \bigg\{ H(X_2,X_1)+H(Y_2,Y_1)-\\
&  ~&& H(X_2,X_1,Y_2,Y_1) + \sum_{\tau=3}^{k} \Big[ H(X_{\tau},X_{\tau-1}) +\\
&  ~&& H(Y_{\tau},Y_{\tau-1}) - H(X_{\tau},X_{\tau-1},Y_{\tau},Y_{\tau-1}) -\\
&  ~&& H(X_{\tau-1}) -  H(Y_{\tau-1}) + H(X_{\tau-1},Y_{\tau-1}) \Big] \bigg\} ~.
\end{alignat*} \normalsize
Note that if the random variables $(X,Y)$, $X$, and $Y$ are higher-order Markov processes, then \eqref{eq:MarkovMI} forms the upper bound on the actual MI \cite{Tan2017}. As $k \rightarrow \infty$, 

\vspace{-0.3cm}\small \begin{alignat*}{2}
I_\text{m}(X;Y) 
&\approx ~&& \frac{1}{k}  \bigg\{ \sum_{\tau=3}^{k} \Big[ H(X_{\tau},X_{\tau-1}) + H(Y_{\tau},Y_{\tau-1}) -\\
&  ~&& \qquad\quad H(X_{\tau},X_{\tau-1},Y_{\tau},Y_{\tau-1}) - H(X_{\tau-1}) -\\
&  ~&& \qquad\quad H(Y_{\tau-1}) + H(X_{\tau-1},Y_{\tau-1}) \Big] \bigg\} ~.
\end{alignat*} \normalsize

It is trivial to see that Proposition \ref{Prop:MinMI} still holds, and that $I_\text{m}(X;Y)$ is minimal when the entropy of $Y$ is minimal. Moreover, when assumptions (a) to (g) hold, then both $I_\text{m}(X;Y^*_{\text{ess}})$ and $I_\text{m}(X;Y^*_{\text{th}})$ are minimal and equal to zero. Now, assume that the Markov processes $(X,Y)$, $X$, and $Y$ are also stationary, \emph{i.e.}, $H(X_1) = H(X_2) =\cdots = H(X_k)$, $H(Y_1) =H(Y_2)=\cdots=H(Y_k)$, $H(X_1,X_2) = H(X_2,X_3)$ $=\cdots = H(X_{k-1},X_k)$, $H(Y_1,Y_2) = H(Y_2,Y_3) =\cdots = H(Y_{k-1},Y_k)$, and that assumptions \ref{as:SysCap1} and \ref{as:SysCap2} are made more realistic as in the i.i.d. case. Then, $I_\text{m}(X;Y_{\text{ess}}) = I_\text{m}(X;Y_{\text{ess}}^*) = 0$, while 

\noindent \small 
\begin{alignat*}{2}
I_\text{m}(X;Y_{\text{th}}) 
&\approx ~&& \frac{g_1(k)}{k} \cdot 0 ~+ \frac{g_2(k)}{k} H(X_{\tau}) ~+ \\
&        ~&& \qquad \frac{g_3(k)}{k} \Big[ H(X_{\tau},X_{\tau-1}) - H(X_{\tau}) \Big] \\ 
&=       ~&& \frac{g_2(k) - g_3(k)}{k} H(X_{\tau}) ~+\\
&        ~&& \frac{g_3(k)}{k} H(X_{\tau},X_{\tau-1}) \,, \quad \tau \in \{2,3,\cdots,k\} ~,
\end{alignat*} \normalsize
where the function $g_1(k)$ gives the number of instances where $(y_{\text{th},\tau}=y_{\text{th},\tau-1}=y^*_{\text{th}})$ or $(y_{\text{th},\tau}=y^*_{\text{th}}, ~y_{\text{th},\tau-1}=x_{\text{th},\tau-1})$, $g_2(k)$ is the number of instances where $(y_{\text{th},\tau}=x_{\text{th},\tau}, ~y_{\text{th},\tau-1}=y^*_{\text{th}})$, $g_3(k)$ is the number of instances where $(y_{\text{th},\tau}=x_{\text{th},\tau}, ~y_{\text{th},\tau-1}=x_{\text{th},\tau-1})$, and $g_1(k)+g_2(k)+g_3(k) = k-2$ \cite{Tan2017}. As $H(X_{\tau}) > 0$, $H(X_{\tau},X_{\tau-1}) > 0$, and $ H(X_{\tau},X_{\tau-1}) > H(X_{\tau})$ (because $X_{\tau}$ and $X_{\tau-1}$ are not perfectly correlated), therefore, $I_\text{m}(X;Y_{\text{th}}) > I_\text{m}(X;Y_{\text{ess}})$. 

\subsection{Privacy Protection for Actual Systems}
For actual systems, the load distributions vary according to the consumer household's state, and their characterisation is the subject of much research. Despite this, consumer privacy is protected if one can achieve a flat grid load that has zero entropy, \emph{i.e.}, zero MI between the sensitive and grid loads. 
While assumptions \ref{as:unreal1} and \ref{as:unreal2} do not hold in reality, it would be possible to implement systems with sufficient storage capacity to average out consumer load (or thermal demand). For ESS-based schemes, one would be able to select $y_{\text{ess}}$ close to $y^*_{\text{ess}}$, given a sufficiently large sample size, as the accuracy of the consumer load sample mean $\hat{X}^{\text{mean}} \rightarrow X^{\text{mean}}$ as $k \rightarrow \infty$. In addition to $X^{\text{mean}}$, the achievable privacy protection of FTL-based UDS schemes is also dependent on $D_{\text{th}}^{\text{mean}}$ and the ratio of $X^{\text{max}}$ to $X^{\text{mean}}$, which are usually fixed and directly affect the number of instances when $y_{\text{th}} = y^*_{\text{th}}$. Note that a larger $X^{\text{max}}$ to $X^{\text{mean}}$ ratio would require a larger $D_{\text{th}}^{\text{mean}}$ to achieve the same level of privacy protection and vice versa. It would be difficult to compare the performance of actual ESS and FTL-based UDS privacy protection schemes, especially since there is a lot of uncertainty in the system parameters for FTLs. Even so, given the analysis above, the additional dependencies of FTL-based schemes (stochastic thermal demand and dependencies on the ambient environment), and the fact that most FTLs are step-loads, properly designed ESS-based schemes should outperform their FTL-based counterparts. 

\section{Formulation of the Optimisation Problem for Numerical Experiments}\label{Sec:ControllerSetup}
Our goal is to compare the performance of privacy protection using ESSs and FTLs in realistic systems to validate our findings from the previous theoretical analysis by simulating a multi-objective model-predictive control-based HEMS controller. For FTLs, we analyse the use of electric hot water heaters (EWHs) and electric resistance space heaters (ERHs), as they better match the analysis in Section \ref{Sec:FTL-ESS_Theory} compared to other FTL types. In this section, the modelling of the ESS and FTLs, the formulation of the privacy objective, and the overall optimisation problems used in the HEMS controllers are presented.

\subsection{Privacy Objective}
There are numerous privacy measures in use across the many different privacy protection schemes available, as briefly summarised by the authors of \cite{Giaconi2018a,Arzamasov2020}. To better match the analysis in Section \ref{Sec:FTL-ESS_Theory}, we adopt a privacy objective function that directly minimises an approximation of \eqref{eq:discMI}. This MI approximate, as proposed in \cite{Chin2017}, assumes that $X$ and $Y$ are i.i.d., and is given by:

\noindent \small 
\begin{alignat}{2} 
I(X;Y)
&\approx   &&\,\,\tilde{I}(X_w;Y_w) \notag \\
&:=      &&\,\sum_{i=1}^m\sum_{j=1}^n \left( a^{ij}_w + \frac{1}{N_{\varepsilon}}\sum_{\tau=w}^{w+W} z^{ij}_{\tau} \right) \times \notag \\
&         &&\left\{ \log \frac{a^{ij}_w}{b^j_w c^i_w} + \frac{\nu}{a^{ij}_w N_{\varepsilon}}\sum_{\tau=w}^{w+W} z^{ij}_{\tau} \right. - \notag \\
&         && \quad\quad\quad\quad\quad  \left. \frac{\nu}{b^j_{\tau} N_{\varepsilon}}\sum_{\tau=w}^{w+W} \sum_{h=1}^m z^{hj}_{\tau} \right\}, \label{eq:linlog}
\end{alignat}\normalsize
at time $w$, where $W+1$ is the prediction horizon, $a^{ij}_w$, $b^j_w$, and $c^i_w$ are constants used in the estimation of the PDFs $p_{X,Y}$, $p_X$ and $p_Y$,  $N_{\varepsilon}$ is the total number of observations used in the estimate, including an additive smoothing constant, $\nu := 1/\log_e 2$, and $z_{\tau}^{ij} \in \{0,1\}$ are binary variables used to estimate the PDFs; see \cite{Chin2017} for details on its derivation. While this MI approximate was shown to directly minimise the MI between $X$ and $Y$, its scalability is limited by the number of binary variables $z_{\tau}^{ij}$, which increases with the prediction horizon length and quantisation levels of $X$ and $Y$. Hence, we relax binary variables $z^{ij}_{\tau}$, \emph{i.e.}, let $z_{\tau}^{ij} \in [0,1]$, in order to make \eqref{eq:linlog} a convex function, and overcome the scalability issues identified in \cite{Chin2017}. This relaxation affects the performance of the controller in terms of minimising MI, but this is outside the scope of this paper. The following constraints are required in the optimisation of \eqref{eq:linlog}:

\noindent \small 
\begin{alignat}{2}
&\sum_{j=1}^{n} z^{i^*j}_{\tau} = 1 ~&& \label{eq:MIstart}\\
& z^{ij}_{\tau} = 0 ~&&, \forall~ i \neq i^*\\
&\sum_{j=1}^{n} z^{i^*j}_{\tau} \bar{y}^{j-1} \leq ~y_{\tau} < \sum_{j=1}^{n} z^{i^*j}_{\tau} \bar{y}^{j}  ~ &&, \label{eq:ZtoY}
\end{alignat} \normalsize
where $i^*$ is the index corresponding to the given value of $x_{\tau}$, $\bar{y}^0 = Y^{\text{min}}$, $\bar{y}^n = Y^{\text{max}}$, and constraint \eqref{eq:ZtoY} links the grid load to its PDF estimate and thus, the MI approximate. 

\subsection{Modelling of an ESS}
Two variables, $P_{c}$ and $P_{d}$, are used to model the instantaneous charging and discharging powers of the ESS, respectively, in order to capture the different loss factors during charge and discharge. Additionally, a binary variable $B_{\textit{ess}}$ is introduced to prevent the simultaneous charging and discharging of the ESS, necessitated by the fact that this is optimal at some time instances due to the privacy objective. While it would be ideal to have a realistic and convex ESS model, its derivation remains an ongoing area of research. Let $E_{\tau}$ be the energy remaining in the ESS at time $\tau$. Then, the following constraints are used to model the ESS in the optimisation problem:

\noindent \small
\begin{alignat}{2}
    &0 \leq P_{c,\tau} \leq B_{\textit{ess},\tau} P^{\text{max}}_{c} &&~ \label{eq:ESSstart}\\
    &0 \leq P_{d,\tau} \leq (1-B_{\textit{ess},\tau}) P^{\text{max}}_{d} &&~\\
    &0 \leq E_{\tau} \leq E^{\text{max}} &&~\\
    &E_{\tau+1} = E_{\tau} + \Delta_t (\eta_{c} P_{c,\tau} - \eta_{d} P_{d,\tau}) &&~\\
    &S_{\tau} = P_{c,\tau} - P_{d,\tau} &&~, \label{eq:ESSend}
\end{alignat} \normalsize
where $\eta_{c}$ and $\eta_{d}$ are the charging and discharging efficiencies of the ESS, respectively, and $\Delta_t$ is the interval of $\tau$.

\subsection{Modelling of an Electric Hot Water Heater}
The thermodynamics in a hot water tank can be modelled by splitting the tank into several sections (nodes). A two-node EWH model proposed in \cite{Jin2014} is adopted in order to better capture the thermodynamics of a real device. As the original model was developed for an electric heat pump, we modify it by replacing the coefficient of performance (COP) with one. Also, we assume a temperature dead-band of $1^{\circ}$C around the temperature set-point. This water heater model is given by the following constraints in the optimisation problem:

\noindent\small
\begin{alignat}{2}
{T}^{low}_{\textit{ewh},\tau+1} &= &&~{T}^{low}_{\textit{ewh},\tau} + \frac{\Delta_t}{C^{low}_{\textit{ewh}}} \left[ \textit{UA}^{low}_{\textit{ewh}} \left(T^{in}_{air,\tau}  - T^{low}_{\textit{ewh},\tau}\right) +   \right. \notag \\
&~  && \left. \Delta m_{hw,\tau} C_p \left(T_{ms} - T^{low}_{\textit{ewh},\tau}\right) + P^{\text{max}}_{\textit{ewh}} U^{low}_{\textit{ewh},\tau} \right] \label{eq:EWHstart}\\
{T}^{up}_{\textit{ewh},\tau+1} &= &&~ {T}^{low}_{\textit{ewh},\tau} + \frac{\Delta_t}{C^{up}_{\textit{ewh}}} \left[ \textit{UA}^{up}_{\textit{ewh}} \left(T^{in}_{air,\tau} - T^{up}_{\textit{ewh},\tau}\right) +   \right. \notag \\
&~  &&~ \left.  \Delta m_{hw,\tau} C_p \left(T^{low}_{\textit{ewh},\tau} ~ - T^{up}_{\textit{ewh},\tau}\right) + P^{\text{max}}_{\textit{ewh}} U^{up}_{\textit{ewh},\tau} \right]
\end{alignat}
\vspace{-0.5cm} 
\begin{align}
& T^{absmin}_{\textit{ewh}} \leq ~T^{up}_{\textit{ewh},\tau} \leq ~T^{absmax}_{\textit{ewh}} \\
& T^{low}_{\textit{ewh},\tau} \leq ~T^{up}_{\textit{ewh},\tau} \\
& U^{low}_{\textit{ewh},\tau} ~+ ~U^{up}_{\textit{ewh},\tau} \leq ~1\\
& S_{\tau} =  \Delta_t(P^{\text{max}}_{\textit{ewh}} U^{low}_{\textit{ewh},\tau} + P^{\text{max}}_{\textit{ewh}} U^{up}_{\textit{ewh},\tau})
\end{align}
\normalsize
where superscripts $low$ and $up$ represent the values for the lower and upper nodes of the tank, respectively. ${T}_{\textit{ewh},\tau}$ is the water temperature of the node, $T^{in}_{air,\tau}$ is the indoor air temperature, $\Delta m_{hw,\tau}$ is the hot water draw, and $U_{\textit{ewh},\tau} \in [0,1]$ is the duty cycle of the EWH tank node at time $\tau$. Also, $C_{\textit{ewh}}$ is the thermal capacitance of the tank node, $\textit{UA}_{\textit{ewh}}$ is the heat loss coefficient of the node, $C_p$ is the heat capacity of water, $T_{ms}$ is the mains water temperature, $P^{\text{max}}_{\textit{ewh}}$ is the rated power of the EWH, $T^{absmin}_{\textit{ewh}}$ is the minimum water temperature required for safety (to mitigate Legionella bacterium growth in pipework), and $T^{absmax}_{\textit{ewh}}$ is the maximum permissible water temperature of the EWH.  Furthermore, to take into account consumer comfort, variables $z^{\textit{comf}}_{\tau} \in \mathbb{R}_{\ge0}$ with constraints:

\noindent\small
\begin{align}
\left(T^{set}_{\textit{ewh}} - 1^\circ\text{C} \right) - T^{low}_{\textit{ewh},\tau} &\leq ~z^{\textit{comf}}_{\tau}\\
T^{up}_{\textit{ewh},\tau} - \left(T^{set}_{\textit{ewh}} + 1^\circ\text{C} \right) &\leq ~z^{\textit{comf}}_{\tau} ~, \label{eq:EWHend}
\end{align} 
\normalsize
are introduced to penalise deviations from consumer set-points for the EWH water $T^{set}_{\textit{ewh}}$.

\subsection{Modelling an Electric Resistance Space Heater}
To model the dynamics of the space heating system, a data-driven model proposed in \cite{Jin2017} is adopted. Similarly, we replace the coefficient of performance (COP) with one, to match the resistance-based ERH. The model coefficients are derived by using statistical learning on data recorded from actual heating systems. The following constraint captures the dynamics of the system: 

\noindent\small
\begin{alignat}{2}
{T}^{in}_{\textit{air},\tau+1} &= &&~{T}^{in}_{\textit{air},\tau} + \gamma_1  \big({T}^{out}_{\textit{air},\tau} -  {T}^{in}_{\textit{air},\tau} \big) + \gamma_2 \big(U_{\textit{erh},\tau} P^{\text{max}}_{\textit{erh}}\big) + \gamma_3 P_{\textit{irr},\tau} \label{eq:ERHstart}
\end{alignat}
\normalsize
where $\gamma_1, ~\gamma_2$, and $\gamma_3$ are parameters learned from data, ${T}^{in}_{\textit{air},\tau}$ and ${T}^{out}_{\textit{air},\tau}$ are the indoor and outdoor temperatures at time $\tau$, respectively, $U_{\textit{erh},\tau} \in [0,1]$ is the ERH duty cycle, $P_{\textit{irr},\tau}$ is the solar irradiance at time $\tau$ ,and $P^{\text{max}}_{\textit{erh}}$ is the rated power of the ERH. Similar to the EWH, the proxy comfort variables $z^{\textit{comf}}_{\tau}$ are used to penalise deviations from consumer set-points. However, as deviations in indoor temperature affect consumer comfort to a higher degree than hot water temperatures, deviations (per $^\circ$C) are penalised with a larger coefficient:

\noindent \small
\begin{alignat}{2}
10\left[\left(T^{set}_{\textit{erh}} - 1^\circ\text{C} \right) - T^{in}_{\textit{air},\tau}\right] &\leq ~z^{\textit{comf}}_{\tau}\\
10\left[T^{in}_{\textit{air},\tau} - \left(T^{set}_{\textit{erh}} + 1^\circ\text{C} \right)\right] &\leq ~z^{\textit{comf}}_{\tau} ~, \label{eq:ERHend}
\end{alignat}
\normalsize
where $T^{set}_{\textit{erh}}$ is the consumer indoor temperature set-point.

\subsection{Optimisation Problem for an ESS-based HEMS Controller}
For an ESS-based HEMS controller, the following objective function is used:

\noindent \small
\begin{equation}
\begin{split} \label{eq:ESSOpt}
& \underset{y_{\tau},z^{ij}_{\tau}}{\mbox{minimise}} \quad \frac{1}{W+1}\sum_{\tau=w}^{w+W} c_{\tau}y_{\tau} + \mu_{w} \tilI(X_{w};Y_{w}) \\
& \mbox{subject to} \quad  (y_{\tau}, ~z^{ij}_{\tau}) \in \mathcal{F}_{\textit{ess},\tau} ~, 
\end{split}
\end{equation}
\normalsize
where $c_{\tau}$ is the cost of energy, $\mu_{w}$ is the price-of-privacy-loss, and the set $\mathcal{F}_{\textit{ess},\tau}$ enforces constraints \eqref{eq:powerbalance}, and \eqref{eq:MIstart} to \eqref{eq:ESSend}.

The inclusion of the energy costs penalises the charging of the ESS during high-price periods, and when coupled with lower prices-of-privacy-loss, discourages multiple charge-discharge cycles within a day. This allows a better comparison with FTL-based systems, which cannot `discharge', and hence have equivalent energy storage capacities limited by the average daily thermal demand and system losses. 

\subsection{Optimisation Problems for FTL-based HEMS Controllers}
In addition to the energy costs, the optimisation objective for FTLs should also minimise consumer comfort violations. We minimise $\|\mathbf{z}^{\textit{comf}}\|_2^2,~ \mathbf{z}^{\textit{comf}} := [{z}^{\textit{comf}}_{w},{z}^{\textit{comf}}_{w+1},\cdots,{z}^{\textit{comf}}_{w+W}]^\top $, which imposes larger penalties for larger comfort violations. Thus, the optimisation problem for an EWH-based HEMS controller is given by

\noindent \small
\begin{equation}
\begin{split} \label{eq:FTLOpt}
& \underset{y_{\tau},z^{ij}_{\tau},\mathbf{z}^{\textit{comf}}}{\mbox{minimise}} \quad \frac{1}{W+1}\sum_{\tau=w}^{w+W} c_{\tau}y_{\tau} + \mu_{w} \tilI(X_{w};Y_{w}) + \rho_{w} \|\mathbf{z}^{\textit{comf}}\|_2^2\\
& \mbox{subject to} \quad  (y_{\tau}, ~z^{ij}_{\tau}, \mathbf{z}^{\textit{comf}}) \in \mathcal{F}_{\textit{th},\tau} ~, \end{split}
\end{equation}
\normalsize
where $\rho_{w}$ is the consumer comfort coefficient, and the set $\mathcal{F}_{\textit{th},\tau}$ enforces the constraints \eqref{eq:powerbalance}, \eqref{eq:MIstart} to \eqref{eq:ZtoY}, and \eqref{eq:EWHstart} to \eqref{eq:EWHend}. For a system with both an EWH and an ERH, set $\mathcal{F}_{\textit{th},\tau}$ in \eqref{eq:FTLOpt} is replaced with the set $\mathcal{F}'_{\textit{th},\tau}$, which now also includes constraints \eqref{eq:ERHstart} to \eqref{eq:ERHend}.

\section{Numerical Experiments}\label{Sec:NumEx}
House $23618$ from the Residential Building Stock Assessment (RBSA) database \cite{RBSA2014} was arbitrarily chosen and used for the numerical simulations. This house is based in Emmett, Idaho, USA, which has a semi-arid climate with cold winters and multiple heating-days. Weather data with 5-minute resolution from Boulder, Colorado, USA, which has a similar climate, was used in the simulations. The HEMS controllers from Section \ref{Sec:ControllerSetup} were simulated for 180 heating-days with hourly resolution in MATLAB 2018a and the Gurobi 8.1.0 optimisation solver.  

For simplicity, we assume that the incoming water supply temperature is constant, and that $\mu_{w}$ and $\rho_{w}$, which can be time-dependent, are also constant. Moreover, for ease of comparison, we assume that the controller has perfect knowledge of the sensitive load across the prediction horizon, and that the models used in the controller accurately represent the actual systems. The equivalent energy storage capacity of an FTL is hard to estimate, depends on many stochastic parameters such as weather conditions and consumer behaviour, and remains an ongoing research challenge. For the simulations, we assumed that this capacity is given by the average daily thermal demand of the household over the simulation period, considering the simulation setup and assumptions. The general simulation parameters are given in Table \ref{tab:GenParam}, while Table \ref{tab:SysParam} gives the system specific parameters. For the FTL-based controllers, $\rho_{w} = 10$.
\begin{table}
\renewcommand{\arraystretch}{1.1}
\centering
\caption{General parameters used in the simulations.}
\label{tab:GenParam}
\begin{tabular}{|l|c|} \hline 
Prediction horizon, $W+1$               & $24$ \\
MI approximate sample size\footnotemark, $N_{\varepsilon}$                       & $201.6$ \\  
Number of $\mathcal{X}$ Bins, $m$       & $24$  \\    
Number of $\mathcal{Y}$ Bins, $n$       & $24$ \\     
Energy Price (peak)                     & $24.6$ cents/kWh\\
Energy Price (off-peak)                 & $13.15$ cents/kWh\\ 
Minimum grid load, $Y^{\text{min}}$     & $0$ kW \\   
Maximum grid load, $Y^{\text{max}}$     & $12$ kW \\
\hline 
\end{tabular} \vspace{-0.2cm}
\end{table}

\footnotetext{including additive smoothing constant}

\begin{table}
\renewcommand{\arraystretch}{1.1}
\centering
\caption{System-specific simulation parameters.}
\label{tab:SysParam}
\begin{tabular}{|l|c|c|c|} \hline
                                    & ESS           & EWH               & ERH \\ \hline
Equivalent storage cap.             & $6.29$ kWh    & $6.29$ kWh        & $32.63$ kWh \\
Power rating                        & $5.5$ kW      & $5.5$ kW          & $4.5$ kW  \\
$1$-way efficiency / COP                      & $96\%$        & $1$               & $1$ \\
Absolute min. temp.                 & -             & $50 ^\circ$C      & - \\
Absolute max. temp.                 & -             & $90 ^\circ$C      & - \\
Consumer set-point                  & -             & $75 ^\circ$C      & $22 ^\circ$C \\
Mains water temp.                   & -             & $10 ^\circ$C      & -\\
Water heat cap., $C_p$              & -             & $4.19$ kJ/K       & -\\
$C_{\textit{ewh}}^{low}$            & -             & $356.15$ kJ/K     & -\\
$C_{\textit{ewh}}^{up}$             & -             & $356.15$ kJ/K     & -\\
Thermal coeff., $\textit{UA}^{low}_{\textit{ewh}}$  & -             & $5.82$e-4 kW/K    & -\\
Thermal coeff., $\textit{UA}^{up}_{\textit{ewh}}$   & -             & $5.82$e-4 kW/K    & -\\
$\gamma_1$                          & -             & -                 & $1.50$e-2 \\
$\gamma_2$                          & -             & -                 & $1.86$e-1\\
$\gamma_3$                          & -             & -                 & $3.45$e-1\\
\hline 
\end{tabular} \vspace{-0.2cm}
\end{table}

The majority of EWHs and ERHs that are currently installed are step loads, while the thermodynamics of the systems are, in reality, continuous. Smaller simulation step sizes would better capture the actual system dynamics, at the expense of computational tractability. Hence, to better match realistic systems and illustrate the mismatch between optimisation models and reality, the continuous duty-cycles from the hourly HEMS controllers were also converted into 5-minute on-off cycles by a secondary controller for system dynamics simulations. This controller attempts to match the HEMS' duty cycle, whilst also enforcing the FTL constraints in Section \ref{Sec:ControllerSetup} at 5-minute resolution. Note that the accuracy of the thermodynamic models is beyond the scope of this paper. To further explore the privacy-protection of both ESS- and FTL-based systems, HEMS controllers that do not consider energy costs were also simulated.

Fig. \ref{fig:GridLoadCurves} shows the load profiles from an ESS-based system and an EWH-based system with discretised control actions (5-minute simulation interval), with $\mu_{w} = 5$, and considering energy costs. As illustrated, the reduced flexibility of the EWH-based system limits its ability to mask the sensitive load, resulting in more instances where the sensitive load is revealed, e.g., around time steps $2866$, $2893$ and $2916$ (highlighted in grey). The ESS is also shown to have a single charge-discharge cycle within 24 hours. Note that here, $H(X;Y)$ is maximised instead as it was impossible to achieve minimal $H(Y)$.

Quantitatively, the privacy leakage of the various systems were assessed by first treating the loads as i.i.d. (IID MI) processes, and then as stationary first-order Markov processes (Markov MI), using the MI estimation methods described in \cite{Chin2018}. It is important to note that the MI estimation methods assume that the FTLs are not privacy-sensitive, \emph{i.e.}, the privacy leakage from the FTL use is not considered. This is particularly important when interpreting the results for $\mu_{w} = 0$ and energy cost is not considered in the objective function. Table \ref{tab:Results} summarises the MI estimates from the various systems. 

Both the ESS- and EWH-based systems reached their privacy protection limit without sacrificing the other objectives with $\mu_w=5$. As seen, the ESS system has less than half the privacy leakage compared to the EWH system with $\mu_w>0$. Without considering energy costs, it can be seen that the ESS achieves much lower MI values (with multiple charge-discharge cycles within a day), while there is only marginal improvement for the EWH due to comfort considerations. With a maximum water draw of $112$ litres within an hour from the $170$ litre hot water tank, there is insufficient flexibility when using the EWH to protect privacy with a $1^\circ$C dead-band. Due to the limited operational flexibility afforded by the hot water tank size and safety considerations, the operation of the EWH does not vary by much given different values of $\mu_{w}$. This can be seen in Fig. \ref{fig:GridLoadCurves_NoECost}, which plots the load curves for different values of $\mu_{w}$, without considering energy cost.

The marginal increase in MI for the ESS without energy costs is due to the binary variables (multiple solution candidates). Moreover, the effects of model mismatch is briefly studied by comparing the 5-minute step load versus non-step load hourly EWH simulations. The actual operation of the EWH differs from the solution of the hourly control actions, as the system dynamics require the secondary controller to make minor adjustments in order to prevent constraint violations (e.g., more accurate water mixing and loss modelling). While one could use models that better represent the continuous dynamics of the thermal system, model mismatch is inevitable in reality, but that is beyond the scope of this paper. The minor adjustments by the secondary controller eventually led to minor reduction in MI in most cases, but that is coincidental.

Even when combining the EWH with an ERH, the privacy protection afforded still falls below that of the ESS with a fraction of the storage capacity for $\mu_w>0$. More importantly, the use of the ERH for privacy entails a significant Markov MI increase, due to the time-correlated dynamics of the system. While there is substantial MI reduction for all systems even with $\mu_w=0$ (the i.i.d. entropy / MI for the sensitive load is $2.710$ bits), if the EWH and ERH usage is privacy-sensitive, then at $\mu_w=0$, the EWH and ERH profiles are unprotected and fully reveal the information contained by their usage.

The limitation of the ERH in providing more privacy protection even when energy costs are ignored, again, lies in the fact that the temperature dead-band is $1^\circ$C, limiting flexibility. This dead-band prevents over-heating the space or letting it cool below comfortable levels. Note that there is very low IID MI when $\mu_w=0$ for the combined EWH and ERH system. This is due to the fact that coincidentally, the period when there is high space heating demand is also the period with high private information leakage (occupied and low-load night periods); and that the ERH usage is assumed to not reveal private information. Fig. \ref{fig:ERHLoadCurve_NoECost} illustrates the sensitive load and grid load curves of the system with an ERH and EWH, without considering energy cost. As shown, while the FTL system is running most of the time, the peak FTL energy demand coincides with the sensitive load troughs, e.g., between time step $2850$ and $2865$. Moreover, given comfort considerations, the controller has limited flexibility in rescheduling the FTL energy demand; as reflected by the cumulative energy consumption of the $\mu_{w}=10$ curve closely matching that of the comfort-only curve within short periods of time.

\begin{figure*}
    \centering
    \includegraphics[trim=8cm 11cm 6cm 11.5cm, clip=true, width=1.80\columnwidth]{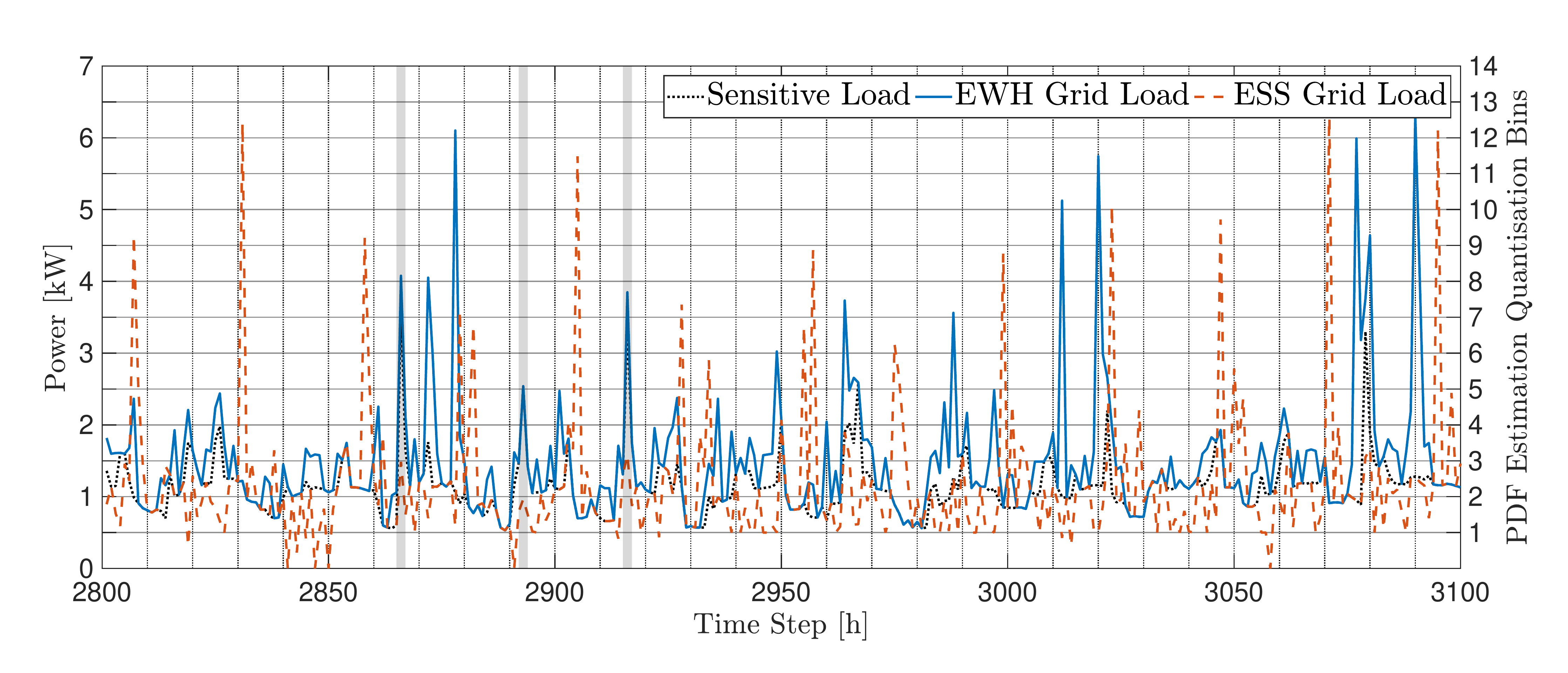}
    \vspace{-0.2cm}
    \caption{The sensitive load, and grid loads with $\mu_w=5$, illustrating how the privacy protection algorithm has shaped the grid visible load to mask the information in the sensitive load, e.g., load peaks and troughs.}
    \label{fig:GridLoadCurves}
\end{figure*}

\begin{figure*}
	\centering
	\includegraphics[trim=8cm 11cm 6cm 11.5cm, clip=true, width=1.80\columnwidth]{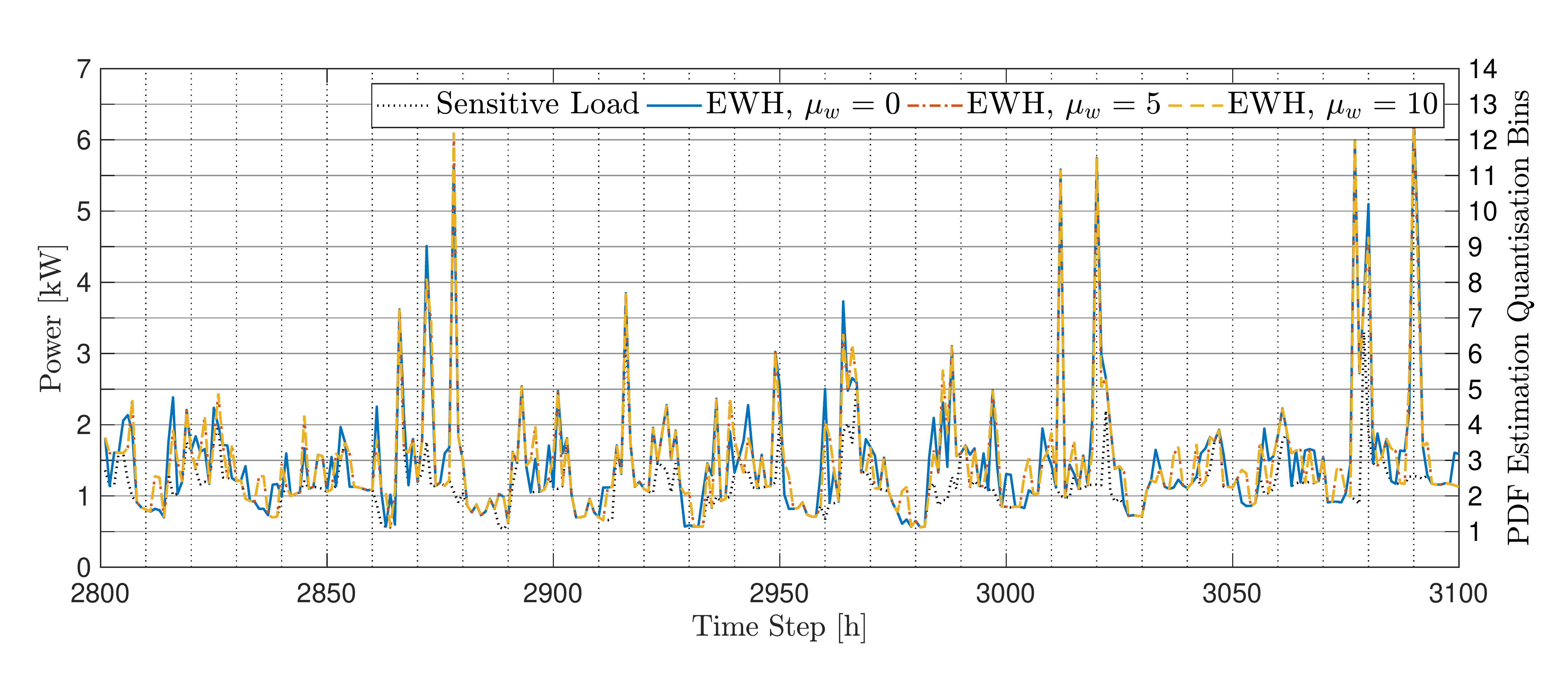}
	\vspace{-0.2cm}
	\caption{The sensitive load, and grid loads for EWH schemes without energy cost and with different $\mu_w$ values, showing that there is a lack of flexibility in the EWH load due to consumer comfort and safety constraints.}
	\label{fig:GridLoadCurves_NoECost}
\end{figure*}

\begin{figure*}
	\centering
	\includegraphics[trim=8cm 11cm 6cm 11.5cm, clip=true, width=1.80\columnwidth]{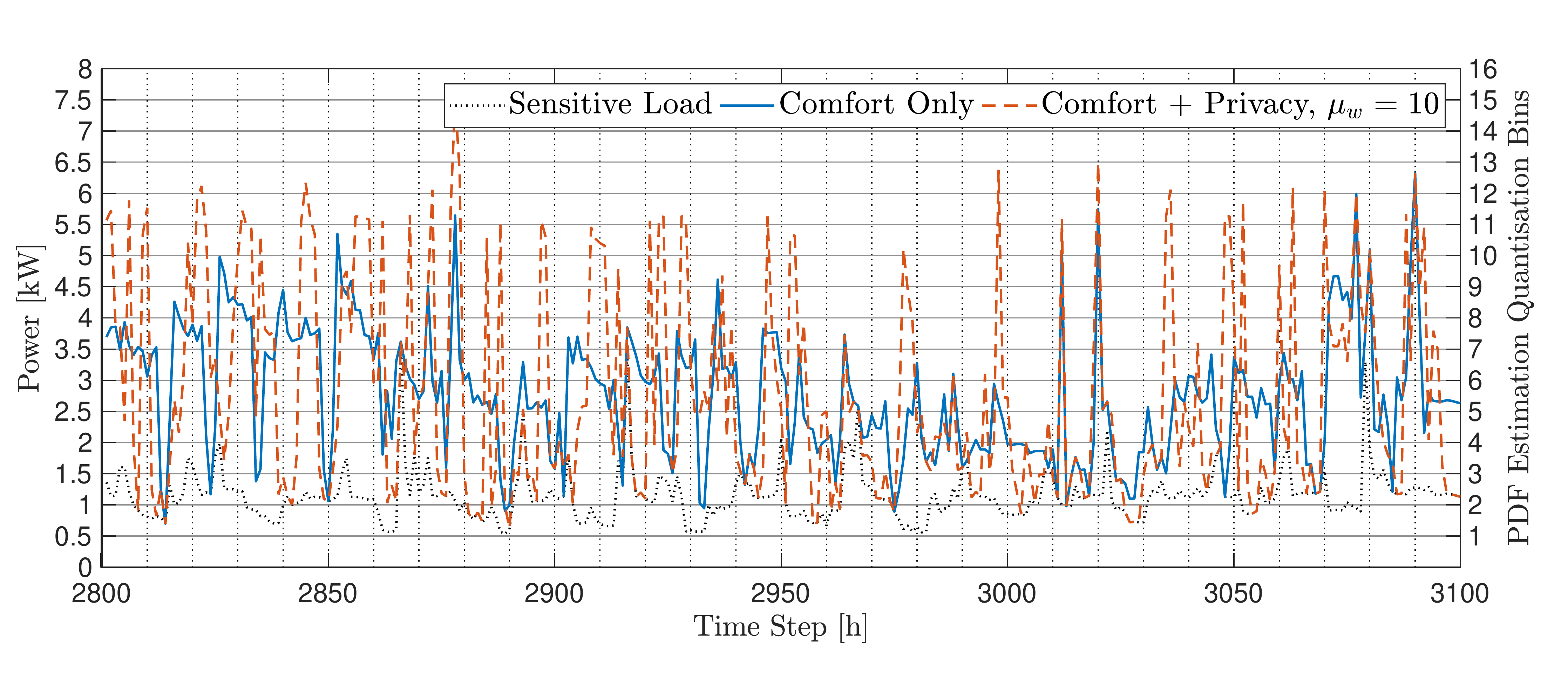}
	\vspace{-0.2cm}
	\caption{The sensitive load, and grid loads for ERH + EWH scheme without energy cost. The load profiles show that the energy consumption of the system when protecting consumer privacy matches that of the ``comfort-only'' setting within a small time window. This illustrates the operational inflexibility of the ERH + EWH system, given the tight temperature dead-bands.}
	\label{fig:ERHLoadCurve_NoECost}
\end{figure*}

\begin{table*}
\renewcommand{\arraystretch}{1.1}
\centering
\caption{Privacy loss of house $23618$ with a 24-hour prediction horizon; illustrating the level of privacy protection afforded by the different systems under different prices-of-privacy-loss.}
\label{tab:Results}
\begin{tabular}{|l|c|c|c|c|c|c|} \hline 
                                            & \multicolumn{2}{c|}{$\mu_{w}=0$}   & \multicolumn{2}{c|}{$\mu_{w}=5$}   & \multicolumn{2}{c|}{$\mu_{w}=10$}\\ \hline
                                            & IID MI    & Markov MI                 & IID MI    & Markov MI                 & IID MI    & Markov MI \\ \hline
ESS with energy costs                       & 0.565     & 0.709                     & 0.286     & 0.678                     & 0.287     & 0.653     \\ \hline
ESS without energy costs                    & -         & -                         & 0.149     & 0.672                     & 0.154     & 0.671     \\ \hline
Step load EWH with energy costs             & 0.656     & 0.859                     & 0.655     & 0.837                     & 0.647     & 0.825     \\ \hline
Step load EWH without energy costs          & 0.657     & 0.831                     & 0.633     & 0.817                     & 0.633     & 0.817     \\ \hline
Non-step load EWH with energy costs         & 0.791     & 0.941                     & 0.693     & 0.870                     & 0.679     & 0.864     \\ \hline
Non-step load EWH without energy costs      & 0.813     & 0.915                     & 0.628     & 0.821                     & 0.628     & 0.824     \\ \hline
Step load EWH and ERH with energy costs     & 0.367     & 1.062                     & 0.362     & 1.066                     & 0.362     & 1.080     \\ \hline
Step load EWH and ERH without energy costs  & 0.136     & 0.788                     & 0.326     & 1.214                     & 0.326     & 1.208     \\ \hline
\end{tabular}\vspace{-0.1cm} 
\end{table*}

In order to get a sense of how an oversized FTL system with more operational flexibility could affect privacy protection, the EWH-based system was simulated with a tank size of $255$ litres instead of $170$ litres ($50\%$ larger). At $\mu_w = 10$ and without considering energy costs, the larger step load EWH-based system achieved slightly better privacy, with an IID MI value of $0.614$ bits (versus $0.633$ bits). This is true, even when considering energy costs with $\mu_w = 10$ ($0.636$ versus $0.647$ bits). Nonetheless, this slight improvement in privacy protection would probably not justify the over-sizing of the FTL systems (compared to an investment in an ESS for cost optimisation and privacy protection).  

As shown in the studies conducted in \cite{Arzamasov2020}, all else being equal, the characteristics of the underlying sensitive load profile affects the level of perceived privacy regardless of privacy metric. To better generalise the findings from the numerical study, the simulations were repeated with the same system parameters for the EWH, ERH and ESS as that used for House $23618$, but using the sensitive load profile of House $21355$ from the RBSA database. The privacy loss for the House $21355$ sensitive load profile, which has an i.i.d. entropy value of $2.246$ bits, is shown in Table \ref{tab:H21355Results}. As can be seen, better privacy protection is achieved for the House $21355$ profile. This is due to the various systems, \emph{i.e.}, EWH, ERH and ESS, being better able to match House $21355$'s lower peak consumption of $4.87$ kW (versus $5.22$ kW for House $23618$), and its lower daily energy consumption. While there is less noticeable improvement for the EWH and ERH-based system relative to one based on an ESS, the findings are similar to those from using the profile of House $23618$. This indicates that the findings from the numerical simulations, which validate the conclusion from the previous theoretical analysis in Section \ref{Sec:FTL-ESS_Theory}, are not specific to a particular sensitive load profile.

\begin{table*}

\centering{
\renewcommand{\arraystretch}{1.1}
\caption{Privacy loss of house $21355$ Using system parameters from house $23618$ and a 24-hour prediction horizon ; illustrating the level of privacy protection afforded by the different systems under different prices-of-privacy-loss.}
\label{tab:H21355Results}
\begin{tabular}{|l|c|c|c|c|c|c|} \hline 
                                            & \multicolumn{2}{c|}{$\mu_{w}=0$}   & \multicolumn{2}{c|}{$\mu_{w}=5$}   & \multicolumn{2}{c|}{$\mu_{w}=10$} \\ \hline
                                            & IID MI    & Markov MI              & IID MI    & Markov MI              & IID MI    & Markov MI   \\ \hline
ESS with energy costs                       & 0.287     & 0.515                  & 0.223     & 0.555                  & 0.232     & 0.539       \\ \hline
ESS without energy costs                    & -         & -                      & 0.140     & 0.564                  & 0.140     & 0.523       \\ \hline
Step load EWH with energy costs             & 0.530     & 0.771                  & 0.517     & 0.766                  & 0.519     & 0.763       \\ \hline
Step load EWH without energy costs          & 0.528     & 0.753                  & 0.496     & 0.738                  & 0.496     & 0.738       \\ \hline
Non-step load EWH with energy costs         & 0.686     & 0.876                  & 0.583     & 0.807                  & 0.563     & 0.802       \\ \hline
Non-step load EWH without energy costs      & 0.710     & 0.844                  & 0.513     & 0.792                  & 0.514     & 0.792       \\ \hline
Step load EWH and ERH with energy costs     & 0.269     & 0.989                  & 0.266     & 0.988                  & 0.266     & 0.997       \\ \hline
Step load EWH and ERH without energy costs  & 0.120     & 0.778                  & 0.271     & 1.070                  & 0.264     & 1.094       \\ \hline
\end{tabular}\vspace{-0.1cm}
}
\end{table*}

Table \ref{tab:ECost} summarises the percentage change in average daily energy cost of the various systems for House $23618$ relative to a no-privacy, non-cost-optimised solution. For the FTL-based systems, the cost basis is taken as the ``consumer comfort-only" setting, while the energy cost of the original sensitive load is used as the basis for the ESS-based systems. These are highlighted in yellow in Table \ref{tab:ECost}. As can be seen, the controller is able to reduce energy costs given a two-tier price tariff, while simultaneously protecting consumer privacy. Nonetheless, the cost savings are minor for the EWH-based systems due to the operational inflexibility of the EWH. For systems that utilise both the EWH and ERH, the cost savings are more significant (both in terms of absolute and relative amounts). This is due to the higher energy demand of the ERH coupled with the higher thermal inertia in the space heating system. For ESS-based systems, the relative cost savings are larger, but the absolute value is still overshadowed by the system's high investment costs. Based on estimates in \cite{EnergySage2020}, which are inline with projections in \cite{IRENA2019}, the latest prices for behind-the-meter battery packs (without installation and balance-of-system costs) are between \$400 and \$750 per kilowatt-hour. This translates to a simple payback period of between $11$ and $20$ years, which is not attractive given the expected lifetime of batteries. On the other hand, prioritising privacy protection (ignoring energy costs) leads to a slight increase in energy cost, which is comparable across both ESS- and FTL-based systems.

\begin{table*}
\centering{
\renewcommand{\arraystretch}{1.1}
\caption{Average daily energy cost of house $23618$ with a 24-hour prediction horizon; showing how the costs change when different systems and prices-of-privacy-loss are used.}
\label{tab:ECost}
\begin{tabular}{|l|c|c|c|} \hline 
                                            & $\mu_{w}=0$   & $\mu_{w}=5$   & $\mu_{w}=10$  \\ \hline
ESS with energy costs                       & -14.96\%      & -14.98\%      & -15.00\%      \\ 
ESS without energy costs \footnotemark      & \cellcolor{yellow}\$4.1693      & +3.280\%      & +4.069\%      \\ \hline
Step load EWH with energy costs             & -0.539\%      & -0.531\%      & -0.538\%      \\ 
Step load EWH without energy costs          & \cellcolor{yellow}\$5.970       & +0.270\%      & +0.270\%       \\ \hline
Non-step load EWH with energy costs         & -1.081\%      & -1.068\%      & -1.055\%      \\ 
Non-step load EWH without energy costs      & \cellcolor{yellow}\$5.957       & +0.102\%      & +0.102\%       \\ \hline
Step load EWH and ERH with energy costs     & -9.086\%      & -9.137\%      & -9.116\%      \\ 
Step load EWH and ERH without energy costs  & \cellcolor{yellow}\$12.202      & +1.903\%      & +2.237\%       \\ \hline
\end{tabular}\vspace{-0.1cm} 
}
\end{table*}

Further experiments are required in order to fully generalise the empirical findings across different load profile characteristics, climate zones, and building and system parameters, which is left as the subject of future work.

\section{Conclusions and Future Outlook}\label{Sec:Conc}
The topic of smart meter consumer privacy is an important one, given that advanced metering infrastructures are often touted as the bedrock of the smart electric grid. Without viable solutions to protect consumer privacy, the deployment of smart meters could potentially be jeopardised. In this paper, we studied the use of resistive flexible thermal-based consumer loads for consumer privacy protection using user demand shaping methods, comparing them against systems using energy storage systems. By conducting a theoretical analysis using mutual information as the quantitative measure of privacy, we show that, based on the fact that flexible thermal-based consumer loads are unable to compensate sensitive load by `discharging', the level of protection afforded by them is below that of energy storage systems. Moreover, as seen from the numerical experiments, the inflexibility of these systems due to the time-specific nature of thermal demand limits their performance; unless one allows for large temperature fluctuations or use largely over-sized systems. Nonetheless, they are still able to afford some level of privacy protection, with incremental energy costs that are comparable to those of energy storage systems, but without the added upfront investment costs. Coupled with their increasing ubiquity in consumer households, research on utilising flexible thermal-based consumer loads in privacy protection schemes, whether standalone or in conjunction with energy storage systems, should be expanded. 

Future work will consider the use of inductive loads and loads with interruptible, but fixed cycle lengths for consumer privacy protection, and the generalisation of the findings using empirical studies.

\footnotetext{The energy cost of the underlying sensitive load is used as the basis for comparing the ESS-based schemes.}

\bibliographystyle{IEEEtran}
\bibliography{./FlexTherm}

\end{document}